\documentclass{article}

\usepackage{authblk}
\usepackage{amsmath,amssymb}
\usepackage{proof}
\usepackage{stmaryrd}
\usepackage{latexsym}
\usepackage{times}
\usepackage{url}
\usepackage{color}
\usepackage{xspace}
\usepackage{tikz}
\usepackage{relsize}
\RequirePackage{txfonts}
\usepackage{microtype}
\usepackage[show]{ed}
\usepackage{lineno}
\usepackage{hyperref}
\usepackage{graphicx}
\usepackage{relsize}
\usepackage{longtable}

\newtheorem{theorem}{Theorem}[section]
\newtheorem{lemma}[theorem]{Lemma}
\newtheorem{proposition}[theorem]{Proposition}

\newenvironment{proof}[1][Proof]{\begin{trivlist}
\item[\hskip \labelsep {\bfseries #1}]}{\end{trivlist}}
\newenvironment{definition}[1][Definition]{\begin{trivlist}
\item[\hskip \labelsep {\bfseries #1}]}{\end{trivlist}}

\newcommand{\qed}{\nobreak \ifvmode \relax \else
      \ifdim\lastskip<1.5em \hskip-\lastskip
      \hskip1.5em plus0em minus0.5em \fi \nobreak
      \vrule height0.75em width0.5em depth0.25em\fi}

\input{mycommands.sty}

\title{Ecumenical modal logic}
\author[1]{Sonia Marin}
\author[2]{Luiz Carlos Pereira}
\author[3]{Elaine Pimentel}
\author[4]{Emerson Sales}

\affil[1]{Department of Computer Science, University College London, UK}
\affil[2]{Philosophy Department, PUC-Rio, Brazil}
\affil[3]{Department of Mathematics, UFRN, Brazil}
\affil[4]{Graduate Program of Applied Mathematics and Statistics, UFRN, Brazil}

\begin{document}
\maketitle
  
  \begin{abstract}
The discussion about how to put together Gentzen's systems for classical and intuitionistic logic in a single unified system is back in fashion. Indeed, recently Prawitz and others have been discussing the so called Ecumenical Systems, where connectives from these logics can co-exist in peace. In Prawitz' system, the classical logician and the intuitionistic logician would share the universal quantifier, conjunction, negation, and the constant for the absurd, but they would each have their own existential quantifier, disjunction, and implication, with different meanings. 
Prawitz' main idea is that these different meanings are given by a semantical framework that can be accepted by both parties. In a recent work, Ecumenical sequent calculi and a nested system were presented, and some very interesting proof theoretical properties of the systems were established. 
In this work we extend Prawitz' Ecumenical idea to alethic $\K$-modalities.
  \end{abstract}

\section{Introduction}\label{sec:intro}
In~\cite{DBLP:journals/Prawitz15} Dag Prawitz proposed a natural deduction system for what was later called {\em Ecumenical logic} (EL),  where classical and intuitionistic logic could coexist in peace. 
In this system, the classical logician and the intuitionistic logician would share the universal quantifier, conjunction, negation, and the constant for the absurd ({\em the neutral connectives}), but they would each have their own existential quantifier, disjunction, and implication, with different meanings. 
Prawitz' main idea is that these different meanings are given by a semantical framework that can be accepted by both parties.  While proof-theoretical
aspects were also considered, his work was more focused on investigating the philosophical significance of the fact that classical logic can be translated into intuitionistic logic.

Pursuing the idea of having a better understanding of Ecumenical systems under the proof-theoretical point of view, in~\cite{Syn20} an Ecumenical sequent calculus ($\LEci$) was proposed.
This enabled not only the proof of some important proof theoretical properties (such as cut-elimination and invertibility of rules), but it also provided a better understanding of the Ecumenical nature of consequence: it is intrinsically intuitionistic, being classical only in the presence of classical succedents. 

Also in~\cite{Syn20} a Kripke style semantics for EL was provided, and it was shown how worlds in the semantics could be adequately interpreted as nestings in nested sequents. 
This also enabled the proposal of an Ecumenical multi-succedent sequent system.

In this work, we propose an extension of EL with the alethic modalities of {\em necessity} and {\em possibility}. There are many choices to be made and many relevant questions to be asked, \eg: what is the ecumenical interpretation of Ecumenical modalities? Should we add classical, intuitionistic, or neutral versions for modal connectives? What is really behind the difference between the classical and intuitionistic notions of truth? 

We propose an answer for these questions in the light of Simpson's meta-logical interpretation of modalities~\cite{Sim94} by embedding the expected semantical behavior of the modal operator into the Ecumenical first order logic.  

We start by highlighting the main proof theoretical aspects of $\LEci$ (Section~\ref{sec:LEci}). This is vital for understanding how the embedding mentioned above will mold the behavior of Ecumenical modalities, since modal connectives are interpreted in first order logics using quantifiers. In Sections~\ref{sec:modal} and~\ref{sec:emodal}, we justify our choices by following closely Simpson's script, with the difference that we prove meta-logical soundness and completeness using proof theoretical methods only. We then provide an axiomatic and semantical interpretation of Ecumenical modalities in Sections~\ref{sec:ax} and~\ref{sec:bi}. This makes it possible to extend the discussion, in Section~\ref{sec:extensions}, to relational systems with the usual restrictions on the relation in the Kripke model. That Section also brings two very interesting observations about intuitionistic $\K\T$. We end the paper with a discussion about logical Ecumenism in general.

\section{The system $\LEci$}\label{sec:LEci}

The language $\Lscr$ used for Ecumenical systems is described as follows. We will use a subscript $c$ for the classical meaning and $i$ for the intuitionistic, dropping such subscripts when formulae/connectives can have either meaning. 

Classical and intuitionistic  n-ary predicate symbols ($P_{c}, P_{i},\ldots$) co-exist in $\Lscr$ but have different meanings. 
The neutral logical connectives $\{\bot,\neg,\wedge,\forall\}$ are common for classical and intuitionistic fragments, while $\{\iimp,\ivee,\iexists\}$ and $\{\cimp,\cvee,\cexists\}$ are restricted to intuitionistic and classical interpretations, respectively. 

The sequent system $\LEci$ was presented in~\cite{Syn20} as the sequent counterpart of Prawitz natural deduction system. 
The rules of $\LEci$ are depicted in Fig.~\ref{Fig:LEci}. Observe that the rules $R_c$ and $L_c$ describe the intended meaning of a classical predicate $P_{c}$ from an intuitionistic predicate $P_i$, 

$\LEci$ has very interesting proof theoretical properties, together with a Kripke semantical interpretation, that allowed the proposal of a variety of ecumenical proof systems, such as a multi-conclusion and a nested sequent systems, as well as several  fragments of such systems~\cite{Syn20}.

\begin{figure}[t]
{\sc Initial and structural rules}
\[
\infer[\init]{A,\Gamma \seq A}{} 
\qquad 
\infer[{\W}]{\Gamma \seq A}{\Gamma \seq \bot}
\]
{\sc Propositional rules}
\[
\infer[{\wedge L}]{A \wedge B,\Gamma \seq 
C}{
A, B,\Gamma \seq C} 
\quad
\infer[{\wedge R}]{\Gamma \seq
A \wedge B}{\Gamma
\seq A \quad \Gamma \seq  B} 
\quad
\infer[{\vee_i L}]{A \vee_i B,\Gamma \seq
C}{A,\Gamma
\seq C \quad B,\Gamma \seq C} 
\quad
\infer[{\vee_i R_j}]{\Gamma \seq  A_1\vee_i A_2}{\Gamma
\seq A_j} 
\]
\[
\infer[\vee_cL]{ A \vee_c B,\Gamma \seq\bot }{A,\Gamma \seq \bot & B, \Gamma
\seq \bot} 
\quad
\infer[\vee_cR]{\Gamma
\seq \ A \vee_c B}{\Gamma, \neg A , \neg B \seq \bot}
\quad
\infer[{\iimp L}]{\Gamma, A\iimp B \seq
C}
{A\iimp B,\Gamma \seq A \quad B, \Gamma  \seq 
C} 
\]
\[
\infer[{\iimp R}]{\Gamma \seq 
A\iimp B}{\Gamma, A \seq  B}
\quad
\infer[\cimp L]{ A \rightarrow_c B,\Gamma \seq\bot }{ A \rightarrow_c B,\Gamma \seq A & B,\Gamma \ \seq
\bot}
\quad
\infer[\cimp R]{\Gamma \seq A
\rightarrow_c B}{\Gamma, A , \neg B \seq  \bot} 
\]
\[
\infer[{\neg L}]{\neg A,\Gamma \seq \bot
}{ \neg A,\Gamma \seq  A}
\quad
\infer[{\neg R}]{\Gamma \seq 
\neg A}{\Gamma, A \seq  \bot}
\quad
\infer[{\bot L}]{\bot,\Gamma \seq A}{}
\quad
\infer[L_c]{ P_c,\Gamma\seq\bot}{P_i,\Gamma\seq \bot } 
\quad
\infer[R_c]{\Gamma \seq P_c}{\Gamma,\neg P_i \seq\bot} 
\]
{\sc Quantifiers}
\[
\infer[\forall L]{\forall x.A,\Gamma \seq C }{A[y/x],\forall x.A,\Gamma \seq C}
\quad
\infer[\forall R]{\Gamma \seq \forall x.A }{\Gamma \seq A[y/x]}
\]
\[
\infer[\exists_i L]{ \exists_ix.A,\Gamma \seq C}{A[y/x],\Gamma \seq C}
\quad
\infer[\exists_i R]{\Gamma \seq \exists_ix.A }{\Gamma \seq A[y/x]}
\quad
\infer[\exists_c L]{\exists_cx.A,\Gamma \seq \bot}{A[y/x],\Gamma \seq  \bot}
\quad
\infer[\exists_c R]{\Gamma \seq \exists_cx.A}{\Gamma, \forall x.\neg A \seq \bot}
\]
\caption{Ecumenical sequent system  $\LEci$. In rules 
$\forall R,\iexists L,\cexists L$, the eigenvariable $y$ is fresh.}\label{Fig:LEci}
\end{figure}

Denoting by $\vdash_{\mathsf{S}} A$ the fact that the formula $A$ is a theorem in the proof system $\mathsf{S}$, the following  theorems are easily provable in $\LEci$:
\begin{enumerate}
\item\label{neg} $\vdash_{\LEci} (A \to_c \bot) \leftrightarrow_i (A \to_i \bot) \leftrightarrow_i (\neg A)$;
\item\label{and} $\vdash_{\LEci}  (A \vee_c B) \leftrightarrow_i  \neg (\neg A \wedge \neg B)$;
\item\label{or} $\vdash_{\LEci} (A \to_c B) \leftrightarrow_i \neg (A \wedge \neg B)$; 
\item\label{foe1} $\vdash_{\LEci} (\exists_cx.A) \leftrightarrow_i \neg (\forall x.\neg A)$.
\end{enumerate}
Note  that~(\ref{neg}) means that the Ecumenical system defined in Fig.~\ref{Fig:LEci} does not distinguish between intuitionistic or classical negations, thus they can be called simply $\neg A$. We prefer to keep the negation operator in the language since the calculi presented in this work make  heavy use of it.

Theorems (\ref{and}) to (\ref{foe1}) are of interest since they relate the classical
and the neutral operators: the classical connectives can be defined using negation, conjunction, and the universal quantifier. 

On the other hand,
\begin{enumerate}
\setcounter{enumi}{4}
\item $\vdash_{\LEci} (A \to_i B) \to_i (A \to_c B)$ but $\not\vdash_{\LEci} (A \to_c B) \to_i (A \to_i B)$ in general;
\item $\vdash_{\LEci} A\vee_c\neg A$ but $\not\vdash_{\LEci} A\vee_i\neg A $ in general;
\item $\vdash_{\LEci} (\neg\neg A) \to_c A $ but  $\not\vdash_{\LEci} (\neg\neg A) \to_i A $ in general;
\item\label{mp} $\vdash_{\LEci} (A\wedge(A  \to_i B)) \to_i B$ but $\not\vdash_{\LEci} (A\wedge(A  \to_c B)) \to_i B $ in general;
\item\label{foe2} $\vdash_{\LEci} \forall x.A \iimp  \neg \exists_cx.\neg A $ but $\not\vdash_{\LEci} \neg\exists_cx.\neg A \to_i \forall x.A$ in general.
\end{enumerate}
Observe that~(\ref{foe1}) and~(\ref{foe2}) reveal the asymmetry between definability of quantifiers: while the classical existential can be defined from the universal quantification, the other way around is not true, in general. This is closely related with the fact that, proving $\forall x.A$ from $\neg\exists_cx.\neg A$ depends on $A$ being a classical formula. We will come back to this in Section~\ref{sec:modal}.

On its turn,  the following result states that logical consequence in $\LEci$ is intrinsically intuitionistic. 
\begin{proposition}[\cite{Syn20}]\label{prop:ev}
$\Gamma\vdash B$ is provable in $\LEci$ iff $\vdash_{\LEci}\bigwedge\Gamma\iimp B$.
\end{proposition}

To preserve  the ``classical behaviour'', \ie, to satisfy all the principles of classical logic \eg\ {\em modus ponens} and the {\em classical reductio}, it is sufficient that the main operator of the formula
be classical (see~\cite{luiz17}). Thus,  ``hybrid" formulas, \ie, formulas that contain
 classical and intuitionistic operators may have a classical behaviour. Formally,
\begin{definition}
A formula $B$ is called \textit{externally classical} (denoted by $B^{c}$) if and only if $B$ is $\bot$, a classical predicate letter, or its
root operator is classical (that is: $\cimp,\vee_c,\exists_c$). A formula $C$ is {\em classical} if it is built from classical atomic predicates using only the connectives: $\cimp,\vee_c,\exists_c,\neg, \wedge,\forall$, and the unit $\bot$.
\end{definition}
For externally classical formulas we can now prove the following theorems
\begin{enumerate}
\setcounter{enumi}{9}
\item $\vdash_{\LEci} (A \to_c B^{c}) \to_i (A \to_i B^{c})$.
\item $\vdash_{\LEci} (A\wedge(A  \to_c B^{c}))  \iimp B^{c} $.
\item $\vdash_{\LEci} \neg\neg B^c\iimp B^c$.
\item $\vdash_{\LEci} \neg\exists_cx.\neg B^c \to_i \forall x.B^c$.
\end{enumerate}
Moreover, notice that all classical right rules as well as the right rules for the neutral connectives in $\LEci$ are invertible. Since invertible rules can be applied eagerly when proving a sequent, this entails that classical formulas can be eagerly decomposed.
As a consequence, the Ecumenical entailment, when restricted to classical succedents (antecedents having an unrestricted form), is classical.
\begin{theorem}[\cite{Syn20}]\label{thm:classical-right}
Let $C$ be a classical formula and $\Gamma$ be a multiset of Ecumenical formulas. Then
$$\vdash_{\LEci}\bigwedge\Gamma\cimp C \mbox{ iff } \vdash_{\LEci}\bigwedge\Gamma\iimp C.$$ 
\end{theorem}
This sums up well, proof theoretically, the {\em ecumenism} of Prawitz' original proposal.

\section{Ecumenical modalities}\label{sec:modal}

In this section we will propose an ecumenical view for alethic modalities. Since there are a number of choices to be made, we will construct our proposal step-by-step.

\subsection{Normal modal logics}\label{sec:nml}
The language of \emph{(propositional, normal) modal formulas} consists of a denumerable set $\Prop$ of propositional symbols and a 
set of propositional connectives enhanced with the unary \emph{modal operators} $\square$ and $\lozenge$ concerning necessity and possibility, respectively~\cite{blackburn_rijke_venema_2001}. 

%
The semantics of modal logics is often determined by means of \emph{Kripke models}. Here, we will follow the approach in~\cite{Sim94}, where a modal logic is characterized by the respective interpretation of the modal model in the meta-theory (called {\em meta-logical characterization}). 

Formally, given a variable $x$, we recall the standard translation $\tradm{\cdot}{x}$ from modal formulas into first-order formulas with at most one free variable, $x$, as follows: if $P$ is atomic, then 
$\tradm{P}{x}=P(x)$; $\tradm{\bot}{x}=\bot$;  for any binary connective $\star$, $\tradm{A\star B}{x}=\tradm{A}{x}\star \tradm{B}{x}$; for the modal connectives
\begin{center}
	\begin{tabular}{lclc@{\quad}lcl}
		$\tradm{\square A}{x}$ & = & $\forall y ( \relfo(x,y) \rightarrow \tradm{A}{y})$ & &
		$\tradm{\lozenge A}{x}$ & = & $\exists y (\relfo(x,y) \wedge \tradm{A}{y})$
		\\
	\end{tabular}
\end{center}
where $\rel(x,y)$ is a binary predicate. 

Opening a parenthesis: such a translation has, as underlying justification, the interpretation of alethic modalities in a Kripke model $\m=(W,R,V)$:
\begin{equation}\label{eq:kripke}
\begin{array}{l@{\qquad}c@{\qquad}l}
\m, w \models \square A & \mbox{ iff } & \text{for all } v \text{ such that  }w \rel v, \m, v \models A.\\
\m, w \models \lozenge A & \mbox{ iff } & \text{there exists } v \text{ such that  } w \rel v	\text{ and } \m, v\models A. 
\end{array}
\end{equation}
$\rel(x,y)$ then represents the \emph{accessibility relation} $R$ in a Kripke frame. This intuition  can be made formal based on the one-to-one correspondence between classical/intuitionistic translations and Kripke modal models~\cite{Sim94}. We close this parenthesis by noting that this justification is only motivational, aiming at introducing modalities. Models will be discussed formally in Section~\ref{sec:bi}.

The object-modal logic OL is then characterized in the first-order meta-logic ML as
$$
\vdash_{OL} A \quad\mbox{ iff } \quad \vdash_{ML} \forall x. \tradm{A}{x}
$$
Hence, if ML is classical logic (CL), the former definition characterizes the classical modal logic $\logick$~\cite{blackburn_rijke_venema_2001}, while if it is intuitionistic logic (IL), then it characterizes the intuitionistic modal logic $\logicik$~\cite{Sim94}. 

In this work, we will adopt EL as the meta-theory (given by the system $\LEci$), hence characterizing what we will defined as the {ecumenical modal logic} $\EK$. 

\subsection{An Ecumenical view of modalities}\label{sec:view}
The language of \emph{Ecumenical modal formulas} consists of a denumerable set $\Prop$ of (Ecumenical) propositional symbols and the set of Ecumenical connectives enhanced with unary \emph{Ecumenical modal operators}. Unlike for the classical case, there is not a canonical definition of constructive or intuitionistic modal logics. Here we 
will mostly follow the approach in~\cite{Sim94} for justifying our choices for the Ecumenical interpretation for {\em possibility} and {\em necessity}.

The ecumenical  translation $\tradme{\cdot}{x}$ from propositional ecumenical formulas into $\LEci$ is defined in the same way as the modal translation $\tradm{\cdot}{x}$ in the last section.
For the case of modal connectives, observe that, due to Proposition~\ref{prop:ev}, the interpretation of ecumenical consequence should be essentially {\em intuitionistic}. With this in mind, the semantical description of $ \square A$ given by~\eqref{eq:kripke} should be understood, for an arbitrary $v$, as: {\em assuming that $wRv$, then $A$ is satisfied in $v$}. Or:
\begin{center}
{\em $A$ is satisfied in $v$ is a consequence of the fact that $wRv$.}
\end{center}
This implies that the box modality is a  {\em neutral connective}. The diamond, on the other hand, has two possible interpretations: classical and intuitionistic, since its leading connective is an existential quantifier. Hence we should have the ecumenical modalities: $\square, \ilozenge, \clozenge$, determined by the translations
\begin{center}
	\begin{tabular}{lclc@{\qquad}lcl}
	$\tradme{\square A}{x}$ & = & $\forall y ( \relfo(x,y) \iimp \tradme{A}{y})$\\
		$\tradme{\ilozenge A}{x}$ & = & $\iexists y (\relfo(x,y) \wedge \tradme{A}{y})$ & &
		$\tradme{\clozenge A}{x}$ & = & $\cexists y (\relfo(x,y) \wedge \tradme{A}{y})$
		\\
	\end{tabular}
\end{center}
Observe that, due to the equivalence~(\ref{foe1}), we have
\begin{enumerate}
\setcounter{enumi}{13}
\item $\clozenge A\iequiv \neg\square\neg A$
\end{enumerate}
On the other hand, $\square$ and $\ilozenge$ are not inter-definable due to (\ref{foe2}).
Finally, if $A^c$ is externally classical, then
\begin{enumerate}
\setcounter{enumi}{14}
\item $\square A^c\iequiv \neg\clozenge\neg A^c$
\end{enumerate}
This means that, when restricted to the classical fragment, $\square$ and $\clozenge$ are duals. 
%
This reflects well the ecumenical nature of the defined modalities. We will denote by $\EK$ the Ecumenical modal logic meta-logically characterized by  $\LEci$ via $\tradme{\cdot}{x}$.

\section{A labeled system for $\EK$}\label{sec:emodal}

\begin{figure}[t]
{\sc Initial and structural rules}
	\[
	\infer[\init]{ x:A,\Gamma \vdash x:A}{}
	\quad
\infer[{\W}]{\Gamma \vdash  x:A}{\Gamma \vdash  y:\bot}
	\]
{\sc Propositional rules}
	\[
	\infer[\wedge L]{ x:A\wedge B,\Gamma \vdash z:C}{ x:A, x:B,\Gamma \vdash z:C}
	\quad
	\infer[\wedge R]{\Gamma \vdash x:A \wedge B}
	{\Gamma \vdash x:A
		&
		\Gamma \vdash x:B}
	\]
	\[
	\infer[\ivee L]{x: A \ivee B,\Gamma \vdash z:C}
	{x:A,\Gamma \vdash z:C
		&
		x:B,\Gamma \vdash z:C}
	\quad
	\infer[\ivee R_j]{\Gamma \vdash x:A_1 \ivee A_2}{\Gamma \vdash x:A_j}
	\]
		\[
	\infer[\cvee L]{x: A \cvee B,\Gamma \vdash x:\bot}
	{x:A,\Gamma \vdash x:\bot
		&
		x:B,\Gamma \vdash x:\bot}
	\quad
	\infer[\cvee R]{\Gamma \vdash x:A \cvee B}{\Gamma,x:\neg A,x:\neg B \vdash x:\bot}
	\]
	\[
	\infer[\iimp L]{x:A \iimp B,\Gamma \vdash z:C}
	{x:A \iimp B,\Gamma \vdash x:A
		&
		 x:B,\Gamma \vdash z:C}
	\quad
	\infer[\iimp R]{\Gamma \vdash x:A \iimp B}{x:A,\Gamma \vdash x:B}
	\]
	\[
	\infer[\cimp L]{x:A \cimp B,\Gamma \vdash x:\bot}
	{x:A \cimp B,\Gamma \vdash x:A
		&
		 x:B,\Gamma \vdash y: \bot}
	\quad
	\infer[\cimp R]{\Gamma \vdash x:A \cimp B}{x:A,x:\neg B,\Gamma \vdash x: \bot}
	\]
	\[
	\infer[{\neg L}]{x:\neg A,\Gamma \vdash  x:\bot
}{x:\neg A,\Gamma \vdash   x:A}
\quad
\infer[{\neg R}]{\Gamma \vdash  
x:\neg A}{x:A,\Gamma \vdash   x:\bot}
\quad
	\infer[\bot]{x:\bot, \Gamma \vdash z:C}{}
\]
	\[
\infer[L_c]{\Gamma, x:P_c\vdash  x:\bot}{\Gamma, x:P_i\vdash  x:\bot } 
\quad
\infer[R_c]{\Gamma \vdash  x:P_c}{\Gamma,x:\neg P_i \vdash  x:\bot} 
\]
{\sc Modal rules}
	\[
	\infer[\square L]{xRy, x:\Box A,\Gamma \vdash z:C}{ xRy,y:A,x:\Box A,\Gamma \vdash z:C}
	\quad
	\infer[\square R]{\Gamma \vdash x:\Box A}{xRy, \Gamma \vdash y:A}
	\quad
	\infer[\ilozenge L]{x:\ilozenge A,\Gamma \vdash z:C}{xRy,  y: A,\Gamma \vdash z:C}
	\]
	\[
	\infer[\ilozenge R]{xRy, \Gamma \vdash x:\ilozenge A}{xRy,\Gamma \vdash y:A}
	\quad
	\infer[\clozenge L]{x:\clozenge A,\Gamma \vdash x:\bot}{xRy,  y: A,\Gamma \vdash x:\bot}
	\quad
	\infer[\clozenge R]{\Gamma \vdash x:\clozenge A}{x:\Box\neg A,\Gamma \vdash x:\bot}
	\]
	\caption{Ecumenical modal system $\labEK$. In rules $\square R,\ilozenge L, \clozenge L$,  the eigenvariable $y$ does not occur free in any formula of the conclusion.}
	\label{fig:EK}
\end{figure}

One of the advantages of having an Ecumenical framework is that some well known classical/intuitionistic systems arise as fragments~\cite{Syn20}. In the following, we will seek for such systems by proposing a labeled sequent system for ecumenical modalities. 

The basic idea behind labeled proof systems for modal logic is to internalize elements of the associated Kripke  semantics (namely, the worlds of a Kripke structure and the accessibility relation between them) into the syntax.
\emph{Labeled modal formulas} 
are either \emph{labeled formulas} of the form $x:A$ or \emph{relational atoms} of the form $x \rel y$, where $x, y$ range over a set of variables and $A$ is a modal formula. 
\emph{Labeled sequents} 
have the form $\Gamma \vdash x:A$, where $\Gamma$ is a multiset containing labeled modal formulas.

Following~\cite{Sim94}, we will prove the following meta-logical soundness and completeness theorem.
\begin{theorem}\label{thm:meta}
Let $\Gamma$ be a multiset of labeled modal formulas and denote $[\Gamma]=\{\rel(x,y)\mid xRy\in\Gamma\}\cup \{\tradme{B}{x}\mid x:B\in\Gamma\}$.
The following are equivalent:
\begin{enumerate}
\item[1.] $\Gamma\vdash x:A$ is provable in $\labEK$. 
\item[2.] $[\Gamma]\seq \tradme{A}{x}$ is provable in $\LEci$.
\end{enumerate}
\end{theorem}
\begin{proof}
We will consider the following translation between $\labEK$ rule applications and $\LEci$ derivations, where the translation for the propositional rules is the trivial one:
\begin{align*}
\vcenter{
\infer[\square L]{xRy,x:\square A,\Gamma\vdash z:C}{xRy,y: A,x:\square A,\Gamma\vdash z:C}
}
&\quad\leadsto\quad
\\
\multispan2{$\hspace*{-2cm}
\vcenter{
\infer=[(\forall L, \iimp L)]{
	\rel(x,y),\forall y.(\rel(x,y)\iimp\tradme{A}{y}),[\Gamma]\seq\tradme{C}{z}
	}{
	\infer{\rel(x,y),\forall y.(\rel(x,y)\iimp\tradme{A}{y}), \rel(x,y)\iimp\tradme{A}{y}, [\Gamma]\seq \rel(x,y)}{}
	&
	\deduce{\rel(x,y),\tradme{A}{y},\forall y.(\rel(x,y)\iimp\tradme{A}{y}),[\Gamma]\seq\tradme{C}{z}}{}&
	}
}
$}
\end{align*}
\begin{align*}
&
\vcenter{
\infer[\square R]{\Gamma\vdash x:\square A}{xRy,\Gamma\vdash y: A}
}
&
\quad\leadsto\quad
&
\vcenter{
\infer=[(\forall R, \iimp R)]{[\Gamma]\seq \forall y.(\rel(x,y)\iimp\tradme{A}{y})}
{\rel(x,y),[\Gamma]\seq \tradme{A}{y}}
}
\\&
\vcenter{
\infer[\lozenge_i L]{x:\ilozenge A,\Gamma\vdash z:C}{xRy,y:A,\Gamma\vdash z:C}
}
&
\quad\leadsto\quad
&
\vcenter{
\infer=[(\iexists L, \wedge L)]{\iexists y.(\rel(x,y)\wedge\tradme{A}{y}),[\Gamma]\seq \tradme{C}{z}}{\rel(x,y),\tradme{A}{y},[\Gamma]\seq \tradme{C}{z}}
}
\\&
\vcenter{
\infer[\lozenge_c L]{x:\clozenge A,\Gamma\vdash x:\bot}{xRy,y:A,\Gamma\vdash x:\bot}
}
&
\quad\leadsto\quad
&
\vcenter{
\infer=[(\cexists L, \wedge L)]{\cexists y.(\rel(x,y)\wedge\tradme{A}{y}),[\Gamma]\seq \bot}
{\rel(x,y),\tradme{A}{y},[\Gamma]\seq \bot}
}
\\&
\vcenter{
\infer[\lozenge_i R]{
	xRy,\Gamma\vdash x:\ilozenge A
	}{
	xRy,\Gamma\vdash y: A}
}
&
\quad\leadsto\quad
&
\vcenter{
\infer=[(\iexists R, \wedge R)]{
		\rel(x,y),[\Gamma]\seq \iexists y.(\rel(x,y)\wedge\tradme{A}{y})
	}{
	\infer{\rel(x,y),[\Gamma]\seq \rel(x,y)}{}
	&
	\deduce{\rel(x,y),[\Gamma]\seq \tradme{A}{y}}{}
	}
}
\\&
\vcenter{
\infer[\lozenge_c R]{\Gamma\vdash x:\clozenge A}{x:\Box\neg A,\Gamma\vdash x:\bot}
}
&
\quad\leadsto\quad
&
\vcenter{
\infer[(\cexists R)]{[\Gamma]\seq \cexists y.(\rel(x,y)\wedge\tradme{A}{y})}
{\forall y.\neg(\rel(x,y)\wedge\tradme{A}{y}),[\Gamma]\seq\bot}
}
\end{align*}
$(1)\Rightarrow(2)$ is then easily proved by induction on a proof of $\Gamma\vdash x:A$ in $\labEK$. 

For proving $(2)\Rightarrow(1)$  observe that
\begin{itemize}
\item the rules $\iimp R,\wedge L, \wedge R$ are invertible in $\LEci$ and $\iimp L$ is semi-invertible on the right (\ie~if its conclusion is valid, so is its right premise);
\item $\vdash_{\LEci}\forall y.\neg (\rel(x,y)\wedge \tradme{A}{y})\leftrightarrow_i
\forall y.\rel(x,y)\iimp\tradme{\neg A}{y}$.
\end{itemize}
Hence, in the translated derivations in $\LEci$ 
{\em provability} is maintained from the end-sequent to the open leaves. 
This means that choosing a formula $\tradme{B}{x}$ to work on is equivalent to performing
all the steps of the translation given above.
Therefore,
any derivation of $[\Gamma]\vdash\tradme{A}{x}$ in $\LEci$ can be transformed into a derivation of the same sequent where all the steps of the translation are actually performed. 
This is, in fact, one of the pillars of the {\em focusing} method ~\cite{andreoli01apal,DBLP:journals/apal/LiangM11}. 
In order to illustrate this, consider
the derivation
\begin{center}
$
\infer[(\forall L)]{\rel(x,y),\tradme{\Box A}{x},[\Gamma]\seq\tradme{C}{z}}
{\deduce{\rel(x,y),\rel(x,y)\iimp\tradme{A}{y},\tradme{\Box A}{x},[\Gamma]\seq\tradme{C}{z}}{\pi}}
$
\end{center}
where one decides to work on the formula $\forall y.(\rel(x,y)\iimp\tradme{A}{y})=\tradme{\Box A}{x}$ obtaining a premise containing the formula $\rel(x,y)\iimp\tradme{A}{y}$, with proof $\pi$. Since $\iimp L$ is semi-invertible on the right and the left premise is straightforwardly provable, 
then $\pi$ can be substituted by the proof:
\begin{center}
$
\infer[\iimp L]{
	\rel(x,y),\rel(x,y)\iimp\tradme{A}{y},\tradme{\Box A}{x},[\Gamma]\seq\tradme{C}{z}
}{
	 \infer{\rel(x,y),\rel(x,y)\iimp\tradme{A}{y},\tradme{\Box A}{x},[\Gamma]\seq\rel(x,y)}{}
	 &
	 \deduce{\rel(x,y),\tradme{A}{y},\tradme{\Box A}{x},[\Gamma]\seq\tradme{C}{z}}{\pi'}
}
$
\end{center}
Thus, by inductive hypothesis, $[xRy,y:A,x:\Box A,\Gamma\vdash z:C]$ is provable in $\labEK$.
\end{proof}
Finally, observe that, when restricted to the intuitionistic and neutral operators, $\labEK$
matches {\em exactly} Simpson's sequent system $\mathcal L_{\Box\Diamond}$~\cite{Sim94}. The analyticity of $\labEK$ is presented in Appendix~\ref{sec:cut}.

\section{Axiomatization}\label{sec:ax}

So far, we have motivated our discussion on Ecumenical modalities based on Simpson's 
approach of meta-logical characterization. But what about the axiomatic characterization of $\EK$? 

Classical modal logic $\K$ is characterized as propositional classical logic, extended with the \emph{necessitation rule} (presented in Hilbert style)
$
A/\square A
$
and the \emph{distributivity axiom}
$
\ka:\;\square(A\ra B)\ra (\square A\ra\square B) 
$. 
Intuitionistic modal logic should then consist of propositional intuitionistic logic plus necessitation and distributivity.
The problem is that there are many variants of axiom $\ka$ that induces classically, but not intuitionistically, equivalent systems.  In fact, the following 
axioms classically follow from $\ka$ and the De Morgan laws, but not in an intuitionistic setting
\begin{center}
$
\begin{array}{lc@{\qquad}l}
\ka_1:\;\square(A\ra B)\ra (\lozenge A\ra\lozenge B)  & & \ka_2:\;\lozenge(A\vee B)\ra (\lozenge A\vee\lozenge B)\\
\ka_3:\;(\lozenge A\ra \square B)\ra \square(A\ra B) & &  \ka_4:\;\lozenge \bot\ra\bot
\end{array}
$
\end{center}
The combination of axiom $\ka$ with axioms $\ka_1$ to $\ka_4$ then exactly characterizes intuitionistic modal logic $\IK$~\cite{plotkin:stirling:86,Sim94}. 

In the ecumenical setting, there are many more variants from $\ka$, depending on the classical or intuitionistic interpretation of the implication and diamond. It is easy to see that 
the intuitionistic versions of the above axioms are provable in $\EK$ (see Appendix~\ref{app}). Hence, by combining this result with cut-elimination (Appendix~\ref{sec:cut}), $\EK$ is {\em complete} w.r.t.~this set of axioms. 
Observe that, since the intuitionistic operators imply the classical ones, if we substitute $i$ by $c$, the resulting clause is either not provable in $\EK$ or it is a consequence of the intuitionistic versions.


Next we will show that $\EK$ is also sound w.r.t.~this set of axioms. For that, we propose an Ecumenical birrelational semantics for $\EK$. The proof passes through a translation from $\labEK$ to $\mathcal L_{\Box\Diamond}$, so we remember that this last labeled system is sound and complete w.r.t. the birelational semantics of $\IK$~\cite{Sim94}.

\section{Ecumenical birelational models}\label{sec:bi}

In~\cite{DBLP:journals/jlp/Avigad01}, the negative translation was used to relate cut-elimination theorems for classical and  intuitionistic logics. Since part of the argumentation was given semantically, a notion of Kripke semantics for classical logic was stated, via the respective semantics for intuitionistic logic and the double negation interpretation (see also~\cite{DBLP:journals/apal/IlikLH10}). In~\cite{luiz17} a similar definition was given, but under the Ecumenical approach, and it was extended to the first-order case in~\cite{Syn20}. We will propose a birelational Kripke semantics for Ecumenical modal logic, which is an extension of the proposal in~\cite{luiz17} to modalities.

\begin{definition}\label{def:ekripke} A {\em birelational Kripke model} is a quadruple  
$\m=(W,\leq,R,V)$ where $(W,R,V)$ is a Kripke model
such that $\wld$ is partially ordered with order $\leq$, the 
satisfaction function $V:\langle W,\leq\rangle\rightarrow  \langle2^\Prop,\subseteq\rangle$ is monotone and:

\noindent
F1. For all worlds $w,v,v'$, if $wRv$ and $v \leq v'$, there is a $w'$ such that $w \leq w'$ and $w'Rv'$;

\noindent
F2. For all worlds $w', w, v$, if $w \leq w'$ and $wRv$, there is a $v'$ such that $w'Rv'$ and $v \leq v'$.

An {\em Ecumenical modal Kripke model} is a birelational Kripke model such that truth of an ecumenical formula at a point $w$ is the
smallest relation $\modelse$ satisfying 

\noindent
$
\begin{array}{l@{\qquad}c@{\qquad}l}
\m,w\modelse P_{i} & \mbox{ iff } & P_{i}\in \val(w);\\
\m,w\modelse A\wedge B & \mbox{ iff } & \m,w\modelse A \mbox{ and } \m,w\modelse B;\\
\m,w\modelse A\vee_i B & \mbox{ iff } & \m,w\modelse A \mbox{ or } \m,w\modelse B;\\
\m,w\modelse A\iimp B & \mbox{ iff } & \text{for all } v \text{ such that  }w \leq v,   
\m,v\modelse A
\mbox{ implies } 
\m,v\modelse B;\\
\m,w\modelse \neg A & \mbox{ iff } & \text{for all } v \text{ such that  }w \leq v,   
\m,v\not\modelse A;\\
\m,w \modelse \bot & & \mbox{never holds};\\
\m, w \modelse \square A & \mbox{ iff } & \text{for all }  v,w' \text{ such that  }w\leq w'\text{ and }  w'\rel v, \m, v \modelse A.\\
\m, w \modelse \ilozenge A & \mbox{ iff } & \text{there exists } v \text{ such that  } w \rel v	\text{ and } \m, v\modelse A. \\
\m,w\modelse P_c  & \mbox{ iff } & \m,w\modelse \neg(\neg P_i);\\
\m,w\modelse A\vee_c B & \mbox{ iff } & \m,w\modelse \neg(\neg A\wedge \neg B);\\
\m,w\modelse A\cimp B & \mbox{ iff } & \m,w\modelse  \neg(A\wedge \neg B).\\
\m, w \modelse \clozenge A & \mbox{ iff } & \m,w\modelse  \neg\square\neg A.
\end{array}
$
\end{definition}
Since, restricted to intuitionistic and neutral connectives, $\modelse$ is the usual birelational interpretation $\models$ for $\IK$ (and, consequently, $\mathcal L_{\Box\Diamond}$ \cite{Sim94}), and since the classical connectives are interpreted 
via the neutral ones using the double-negation translation, an Ecumenical modal Kripke model is nothing else than the standard birelational Kripke model for intuitionistic modal logic $\IK$.
Hence, it is not hard to prove soundness and completeness of the semantical interpretation above w.r.t. the sequent system $\labEK$.
We start by defining a translation from $\labEK$ to $\mathcal L_{\Box\Diamond}$. We will abuse the notation and represent the connectives of $\IK/\mathcal L_{\Box\Diamond}$ using the neutral/intuitionistic correspondents in $\EK/\labEK$.
\begin{definition}\label{def:transS}
Let $\tradmk{\cdot}$ be the translation between formulas in $\EK$ and $\IK$ recursively defined as

\noindent
$
\begin{array}{lclc@{\qquad}lcl}
\tradmk{P_i} & = & P_i & & \tradmk{P_c} & = & \neg(\neg(P_i))\\
\tradmk{\bot} & = & \bot & &
\tradmk{\neg A} & = & \neg \tradmk{A}\\
\tradmk{A\wedge B} & = & \tradmk{A}\wedge\tradmk{B} & &
\tradmk{A\ivee B} & = & \tradmk{A}\ivee\tradmk{B}\\
\tradmk{A\iimp B} & = & \tradmk{A}\iimp\tradmk{B} & &
\tradmk{A\cvee B} & = & \neg(\neg\tradmk{A}\wedge\neg\tradmk{B})\\
\tradmk{A\cimp B} & = & \neg(\tradmk{A}\wedge\neg\tradmk{B})& &
\tradmk{\Box A} & = & \Box \tradmk{A}\\
\tradmk{\ilozenge A} & = & \ilozenge \tradmk{A} & &
\tradmk{\clozenge A} & = & \neg\Box\neg \tradmk{A}\\
\end{array}
$\\

\noindent
The translation $\tradms{\cdot}:\labEK\to \mathcal L_{\Box\Diamond}$ is defined as
$\tradms{x:A}=x:\tradmk{A}$ and assumed identical on relational atoms.
\end{definition}
Since the translations above preserve the double-negation interpretation of classical connectives into intuitionistic (modal) logic, it is possible to prove:
\begin{lemma}\label{lem:trad1}
$\vdash_{\labEK} \Gamma\vdash x:A$
iff $\vdash_{\labEK}\tradms{\Gamma\vdash x:A}$
iff $\vdash_{\mathcal L_{\Box\Diamond}}\tradms{\Gamma\vdash x:A}$.
\end{lemma}

\begin{lemma}\label{lem:trad3} Let $A$ be a formula in $\EK$. Then 
$\modelse A$ iff $\models \tradmk{A}$.
\end{lemma}
\begin{proof} 
First of all, note that an Ecumenical modal Kripke model is totally determined by the valuation of the (intuitionistic) propositional variables. 
The proof follows then by easy structural induction on the formula $A$. 
For example, if $A=\clozenge B$ then 
$\modelse A$ iff $\modelse \neg\Box\neg B$. By inductive hypothesis, $\modelse B$ iff $\models \tradmk{B}$. Since $\neg$ and $\Box$ are neutral operators, then 
$\modelse\neg\Box\neg B$ and  iff $\models \neg\Box\neg \tradmk{B}$.
\end{proof}
Now, observe that every formula in $\IK$ is a formula in $\EK$, hence the following theorem holds.
\begin{theorem}\label{thm:scLEci}
$\EK$ is sound and complete w.r.t.~the Ecumenical modal Kripke semantics, that is, $\vdash_{\EK} A$ iff $\modelse A$.
\end{theorem}
%
Moreover, $\EK$ is sound and complete w.r.t. the axioms $\ka-\ka_4$ presented in Section~\ref{sec:ax}.

\section{Extensions}\label{sec:extensions}
Depending on the application, several further modal logics can be defined as extensions of $\K$ by simply restricting the class of frames we consider. 
Many of the restrictions one can be interested in are definable as formulas of first-order logic, where the binary predicate $\rel(x,y)$ refers to the corresponding accessibility relation. 
Table \ref{tab:axioms-FOconditions} summarizes some of the most common logics, the corresponding frame property, together with the modal axiom capturing it~\cite{Sah75}. 

In the Ecumenical setting, we adopt the motto that ``relational predicates are second order citizens'' suggested in Section~\ref{sec:view} and interpret the implications in the axioms intuitionisticaly. We will refer to the ecumenical logic satisfying the axioms $F_1,\ldots, F_n$ as $\EK F_1\ldots F_n$.

We conjecture that the semantics of a given logic $\EK F_1 \ldots F_n$ can be inferred from the one for $\EK$ of Definition $\ref{def:ekripke}$: We just consider Ecumenical models whose accessibility relation satisfies the set of properties $\{F_1, \ldots, F_n\}$ in place of generic Ecumenical models. 

\begin{table}[t]
\begin{center}
  \begin{tabular}{|c|c|c|}
  \hline
    \textbf{Axiom} & \textbf{Condition} & \textbf{First-Order Formula} \\
  \hline
    $\T:\,\square A \rightarrow A \wedge A \rightarrow \lozenge A$ & Reflexivity & $\forall x. \rel(x,x)$ \\
\hline
    $\4:\,\square A \rightarrow \square \square A \wedge \lozenge\lozenge A \rightarrow \lozenge A$ & Transitivity & $\forall x,y,z. (\rel(x,y) \wedge \rel(y,z)) \rightarrow \rel(x,z)$ \\
\hline
    $\5:\,\square A \rightarrow \square \lozenge A \wedge \lozenge\square A \rightarrow \lozenge A$ & Euclideaness & $\forall x,y,z. (\rel(x,y) \wedge \rel(x,z)) \rightarrow \rel(y,z)$ \\
\hline
    $\B:\,A \rightarrow \square \lozenge A \wedge\lozenge \square A \rightarrow A$ & Symmetry & $\forall x,y. \rel(x,y) \rightarrow \rel(y,x)$ \\
\hline
  \end{tabular}
  \end{center}
  \caption{Axioms and corresponding first-order conditions on $\rel$.}
    \label{tab:axioms-FOconditions}  
\end{table}

Furthermore, following the approaches in~\cite{Sim94,Vig00,Neg05}, we can transform the axioms in Table~\ref{tab:axioms-FOconditions} into rules. 
%
In future work, we would like to investigate how these rules behave w.r.t~the ecumenical setting.

\begin{figure}[t]
\[
\infer[\T]{\Gamma\vdash w:C}{xRx,\Gamma\vdash w:C}
\quad
\infer[\4]{xRy,yRz,\Gamma\vdash w:C}{xRz,xRy,yRz,\Gamma\vdash w:C}
\]\[
\infer[\5]{xRy,xRz,\Gamma\vdash w:C}{yRz,xRy,xRz,\Gamma\vdash w:C}
\quad
\infer[\B]{xRy,\Gamma\vdash w:C}{yRx,xRy,\Gamma\vdash w:C}
\]
\caption{Labeled sequent rules corresponding to axioms in Table~\ref{tab:axioms-FOconditions}.  
}
\label{fig:extK}
\end{figure}

As a first step in this research direction, we finish this section with two very interesting observations about the case of axiom $\T$, illustrating the complexity of the interaction of modal axioms and ecumenical connective. 
First of all, recall~\cite{Sim94,DBLP:conf/fossacs/Strassburger13} that by itself, $\square A \iimp A$ does not enforce reflexivity of an intuitionistic model. In fact, in frames having the reflexivity property, both $\square A \iimp A$ and $A \iimp \ilozenge A$ are provable.
\[
\infer[\wedge R]{xRx \vdash x: (\square A\iimp A) \wedge (A\iimp \ilozenge A)}
{\infer[\iimp R]{xRx \vdash x:\square A\iimp A}{\infer[\square L]{xRx, x:\square A \vdash x: A}{\infer[\init]{xRx, x: A \vdash x: A}{}}}&\infer[\iimp R]{xRx \vdash x: A \iimp \ilozenge A}{\infer[\ilozenge R]{xRx, x: A \vdash x: \ilozenge A}{\infer[\init]{xRx, x: A \vdash x: A}{}}}}
\]
For the converse, since $\square$ and $\ilozenge$ are not inter-definable, we need to add $A \iimp \ilozenge A$ in order to still be complete w.r.t. reflexive models. 

Finally, it is well known that the set intuitionistic propositional operators is ``independent'', \ie, that each operator cannot be defined in terms of the others (Prawitz proposed a syntactical proof of this result in~\cite{prawitz1965}). It is also known that in the case of (many) constructive modal logics the modal operators $\Box$ and $\lozenge$ are independent of each other, as it is the case in $\EK$. But what would be the consequence of adding $\neg\ilozenge\neg A\iimp \Box A$ as an extra axiom to $\EK$? The following derivation shows that the addition of this new axiom has a disastrous propositional consequence.
$$\small
\infer[\cut]{\vdash x:(A\ivee\neg A)}
{\infer[eq]{\vdash x:\Box (A\ivee\neg A)}
{\infer[\neg R]{\vdash x:\neg\ilozenge\neg (A\ivee\neg A)}
{\infer[\ilozenge L]{x:\ilozenge\neg (A\ivee\neg A)\vdash x:\bot}
{\infer=[\neg L,\ivee R2,\neg R]{xRy,y:\neg (A\ivee\neg A)\vdash x:\bot}
{\infer[\neg L,\ivee R1]{xRy,y:A, y:\neg (A\ivee\neg A)\vdash y:\bot}
{\infer[\init]{xRy,y:A, y:\neg (A\ivee\neg A)\vdash y:A}{}}}}}}&
\infer[\T]{x:\Box (A\ivee\neg A)\vdash x:(A\ivee\neg A)}
{\infer[\Box L]{xRx,x:\Box (A\ivee\neg A)\vdash x:(A\ivee\neg A)}
{\infer[\init]{xRx,x: (A\ivee\neg A)\vdash x:(A\ivee\neg A)}{}}}}
$$
where $eq$ represents the proof steps of the substitution of a boxed formula for its 
diamond version.\footnote{We have presented a proof with $\cut$ for clarity, remember that $\labEK$ has the cut-elimination property (see Appendix~\ref{sec:cut}).} 

That is, if $\Box$ and $\lozenge_i$ are inter-definable, then $A\vee_i\neg A$ is a theorem and intuitionistic $\K\T$ collapses to a classical system!

\section{Discussion and conclusion}\label{sec:conc}
Some questions naturally arise with respect to Ecumenical systems: what (really) are Ecumenical systems? What are they good for? Why should anyone be interested in Ecumenical systems? What is the real motivation behind the definition and development of Ecumenical systems? Based on the specific case of the Ecumenical system that puts classical logic and intuitionist logic coexisting in peace in the same codification, we would like to propose three possible motivations for the definition, study and development of Ecumenical systems.
\paragraph{Philosophical motivation}
This was the motivation of Prawitz. Inferentialism, and in particular, logical inferentialism, is the semantical approach according to which the meaning of the logical constants can be specified by the rules that determine their correct use. According to Prawitz~\cite{DBLP:journals/Prawitz15},
\begin{quote}
{\em ``Gentzen's introduction rules, taken as meaning constitutive of the logical constants of the language of predicate logic, agree, as is well known, with how intuitionistic mathematicians use the constants. On the one hand, the elimination rules stated by Gentzen become all justified when the constants are so understood because of there being reductions, originally introduced in the process of normalizing natural deductions, which applied to proofs terminating with an application of elimination rules give canonical proofs of the conclusion in question. On the other hand, no canonical proof of an arbitrarily chosen instance of the law of the excluded middle is known, nor any reduction that applied to a proof terminating with an application of the classical form of reductio ad absurdum gives a canonical proof of the conclusion.''}
\end{quote}
But what about the use classical mathematicians make of the logical constants? Again, according to Prawitz,
\begin{quote}
{\em ``What is then to be said about the negative thesis that no coherent meaning can be attached on the classical use of the logical constants? Gentzen's introduction rules are of course accepted also in classical reasoning, but some of them cannot be seen as introduction rules, that is they cannot serve as explanations of meaning. The classical understanding of disjunction is not such that $A\vee B$ may be rightly asserted only if it is possible to prove either A or B, and hence Gentzen's introduction rule for disjunction does not determine the meaning of classical disjunction.''}
\end{quote}
As an alternative, in a recent paper~\cite{Murzi2018} Murzi presents a different approach to the extension of inferentialism to classical logic. There are some natural (proof-theoretical) inferentialist requirements on admissible logical rules, such as harmony and separability (although harmonic, Prawitz' rules for the classical operators do not satisfy separability). According to Murzi, our usual logical practice does not seem to allow for an inferentialist account of classical logic (unlike what happens with respect to intuitionistic logic). Murzi proposes a new set of rules for classical logical operators based on: absurdity as a punctuation mark, and Higher-level rules~\cite{DBLP:journals/sLogica/Schroeder-Heister14}. This allows for a ``pure'' logical system, where negation is not used in premises. 
\paragraph{Mathematical/computational motivation}
(This was actually the original motivation for proposing Ecumenical systems.) The first Ecumenical system (as far as we know) was defined by Krauss  in a technical report of the University of Kassel~\cite{Krauss} (the text was never published in a journal). The paper is divided in two parts: in the first part, Krauss' Ecumenical system is defined and some  properties proved. In the second part, some theorems of basic algebraic number theory are revised in the light of this (Ecumenical) system, where constructive proofs of some ``familiar classical proofs'' are given (like the proof of Dirichlet's Unit Theorem).
The same motivation can be found in the final passages of the paper~\cite{DBLP:journals/Dowek16a}, where Dowek examines what would happen in the case of axiomatizations of mathematics. Dowek gives a simple example from Set Theory, and ends the paper with this very interesting remark:
\begin{quote}
{\em ``Which mathematical results have a classical formulation that can be proved from the axioms of constructive set theory or constructive type theory and which require a classical formulation of these axioms and a classical notion of entailment remains to be investigated.''}
\end{quote}
\paragraph{Logical motivation}
In a certain sense, the logical motivation naturally combines certain aspects of the philosophical motivation with certain aspects of the mathematical motivation.
According to Prawitz, one can consider the so-called classical first order logic as {\em ``an attempted codification of a fragment of inferences occurring in [our] actual deductive practice''}. Given that there exist different and even divergent attempts to codify our (informal) deductive practice, it is more than natural to ask about what relations are entertained between these codifications. 
Ecumenical systems may help us to have a better understanding of the relation between classical logic and intuitionistic logic.
But one could say that, from a logical point of view, there's nothing new in the ecumenical proposal: Based on translations, the new classical operators could be easily introduced by ``explicit definitions''. Let us consider the following dialogue between a classical logician (CL) and an intuitionistic logician (IL), a dialogue that may arise as a consequence of the translations mentioned above:
\begin{itemize}
\item IL: if what you mean by $(A \vee B)$ is $\neg (\neg A \wedge \neg B)$, then I can accept the validity of $(A \vee \neg A)$!
\item CL: but I do not mean $\neg (\neg A \wedge \neg \neg A)$  by $(A \vee \neg A)$. One must distinguish the excluded-middle from the the principle of non-contradiction. When I say that Goldbach's conjecture is either true or false, I am not saying that it would be contradictory to assert that it is not true and that it is not the case that it is not true!
\item IL: but you must realize that, at the end of the day, you just have one logical operator, the Sheffer stroke (or the Quine's dagger).
\item CL: But this is not at all true! The fact that we can define one operator in terms of other operators does not imply that we don't have different operators! We do have 16 binary propositional operators (functions). It is also true that we can prove $\vdash (A \vee_{c} B) \leftrightarrow \neg(\neg A \wedge \neg B)$ in the ecumenical system, but this doest not mean that we don't have three different operators, $\neg $, $\vee_{c}$ and $\wedge $.
\end{itemize}
Maybe we can resume the logical motivation in the following (very simple) sentence:
\begin{quote}
Ecumenical systems constitute a new and promising instrument to study the nature of different (maybe divergent!) logics.
\end{quote}
Now, what can we say about {\em modal} Ecumenical systems?
Regarding the {\bf philosophical view}, in~\cite{Syn20} we have used invertibility results in order to obtain a sequent system for Prawitz' Ecumenical logic with a minimal occurrences of negations, moving then towards 
a ``purer'' Ecumenical system. Nevertheless,  negation still plays an important r\^{o}le on interpreting classical connectives.  This is transferred to our definition of Ecumenical modalities, where the classical possibility is interpreted using negation. We plan to investigate what would be the meaning of classical possibility without impure rules. For the {\bf mathematical view}, our use of intuitionistic/classical/neutral connectives allows for a more chirurgical detection of the parts of a mathematical proof that are intrinsically intuitionistic, classical or independent. We now bring this discussion to modalities. Finally, concerning the {\bf logical view}, it would be interesting to explore some relations between general results on translations and Ecumenical systems, expanding this discussion to modalities. 

To finish, we would like to say a word about our choices.  It seems to be a common view, in the proof theory community, that Simpson's view is the more reasonable approach for modalities and intuitionism.
From that, the choice of a labeled proof system for $\EK$ seems only natural. But labeled systems have a very unfortunate feature: it is really tricky
to define an interpretation of sequents into the logical language. This problem often disappears when moving to nested-like systems~\cite{Brunnler:2009kx,Poggiolesi:2009vn,DBLP:conf/fossacs/Strassburger13,Lellmann:2015lns},  since the nestings keep the tree-structure information, matching exactly the history of a backwards proof search in an ordinary sequent calculus. 
Also, having an Ecumenical nested system would most probably allow for a comparison, in one system, between the nested sequent for $\IK$~\cite{DBLP:conf/fossacs/Strassburger13} and for $\CK$~\cite{DBLP:journals/corr/ArisakaDS15,DBLP:journals/aml/KuznetsS19}. Hence this is a path worth pursuing, together with the comparison of $\labEK$ with other labelled sequent systems for intuitionistic modal logics, specially the recent ones proposed in~\cite{marin:hal-02390454} and~\cite{DBLP:journals/corr/abs-1901-09812}.

%
%
%
%

\newpage
\appendix
\section{Cut-elimination for $\labEK$}\label{sec:cut}
In face of Theorem~\ref{thm:meta} most of the proof theoretical properties of the system $\labEK$ can be inherited from $\LEci$. It is not different for the property of cut-elimination. Hence we will only illustrate the process here.

The extension of the Ecumenical weight for formulas presented~\cite{luiz17} to modalities is defined bellow.
\begin{definition}\label{def:ew}
The Ecumenical weight ($\ew$) of a formula in $\Lscr$ is recursively  defined as
\begin{itemize}
\item $\ew(P_i)=\ew(\bot)=0$; 
\item $\ew(A\star B)=\ew(A)+\ew(B)+1$ if $\star\in\{\wedge,\iimp,\ivee\}$;
\item $\ew(\heartsuit A)=\ew(A)+1$ if $\heartsuit\in\{\neg,\ilozenge,\Box\}$;
\item $\ew(A\circ B)=\ew(A)+\ew(B)+4$ if $\circ\in\{\cimp,\cvee\}$;
\item $\ew(P_c)=4$;
\item $\ew(\clozenge A)=\ew(A)+4$.
\end{itemize}
\end{definition}
Intuitively, the Ecumenical weight measures the amount of extra information needed (the negations added) in order to define the classical connectives from the intuitionistic and neutral ones.
\begin{theorem}\label{thm:cut}
The rule
\[
\infer[\cut]{\Gamma\vdash z:C}
{\Gamma\vdash x:A & x:A,\Gamma\vdash z:C}
\]
is admissible in $\labEK$.
\end{theorem}
\begin{proof}
The proof is by the usual Gentzen method. The principal cases either eliminate the top-most cut or substitute it for cuts over simpler ecumenical formulas, \eg\ 
\[
\vcenter{\infer[\cut]{\Gamma\vdash z:\bot}
{\infer[\clozenge R]{\Gamma\vdash x:\clozenge A}
{\deduce{x:\Box\neg A,\Gamma\vdash x:\bot}{\pi_1}}&
\infer[\clozenge L]{x:\clozenge A,\Gamma, \vdash z:\bot}
{\deduce{xRy, y: A,\Gamma\vdash x:\bot}{\pi_2}}}}
\leadsto
\]
\[
\infer[\cut]{\Gamma\vdash z:\bot}
{\infer[\Box R]{\Gamma\vdash x:\Box\neg A}
{\infer[\neg R]{xRy,\Gamma\vdash y:\neg A}
{\deduce{xRy,y:A,\Gamma\vdash y:\bot}{\pi_2}}}&
\deduce{ x:\Box\neg A,\Gamma, \vdash z:\bot}{\pi_1}}
\]
Observe that the label of bottom is irrelevant due to the weakening rule $\W$ (that we have suppressed). Hence the Ecumenical weight on the cut formula passes from $\ew( \clozenge x.A)=\ew(A)+4$ to $\ew(\Box\neg A)=\ew(A)+2$.

The non-principal cuts can be flipped up as usual, generating cuts with smaller cut-height.
\end{proof}

\section{Proofs of axioms  $\ka-\ka_4$ in $\labEK$}\label{app}

We show next that axioms $\ka-\ka_4$ are provable in $\labEK$.
\begin{itemize}
\item $\ka:\square(A\iimp B) \iimp (\square A \iimp \square B) $:
     \[
     \infer[\iimp R]{\vdash x:\square(A\iimp B) \iimp (\square A \iimp \square B) }
     {\infer[\iimp R]{x: \square(A\iimp B) \vdash x: \square A \iimp \square B}
     {\infer[\square R]{x: \square(A\iimp B), x: \square A \vdash x: \square B}
     {\infer[\square L]{xRy, x: \square(A\iimp B), x: \square A \vdash y: B}
     {\infer[\square L]{xRy, x: \square(A\iimp B), y: A ,x: \square A \vdash y: B}
     {\infer[\iimp L]{xRy, y: A\iimp B, x: \square(A\iimp B), y: A ,x: \square A \vdash y: B}
     {\infer[]{xRy, y: A\iimp B, x: \square(A\iimp B), y: A ,x: \square A \vdash y: A}{}& \infer[]{xRy, x: \square(A\iimp B), y: A ,x: \square A, y: B \vdash y: B}{}
     }}}}}}
     \]
\item $\ka_1:\square(A\iimp B) \iimp (\ilozenge A \iimp \ilozenge B)$
    \[
    \infer[\iimp R]{\vdash x: \square(A\iimp B) \iimp (\ilozenge A \iimp \ilozenge B)}
    {\infer[\iimp R]{x: \square(A\iimp B) \vdash x: \ilozenge A \iimp \ilozenge B}
    {\infer[\ilozenge L]{x: \square(A\iimp B), x: \ilozenge A \vdash x: \ilozenge B}
    {\infer[\ilozenge R]{xRy, x: \square(A\iimp B), y: A \vdash x: \ilozenge B}
    {\infer[\square L]{xRy, x: \square(A\iimp B), y: A \vdash y: B}
    {\infer[\iimp L]{xRy, y: A\iimp B, x: \square(A\iimp B), y: A \vdash y: B}
    {\infer[]{xRy, y: A\iimp B, x: \square(A\iimp B), y: A \vdash y: A}{}& \infer[]{xRy, x: \square(A\iimp B), y: A, y: B \vdash y: B}{}
    }}}}}}
    \]
\item $\ka_2:\ilozenge(A\vee_i B)\iimp (\ilozenge A \vee_i \ilozenge B)$
    \[
    \infer[\iimp R]{\vdash x: \ilozenge(A\vee_i B)\iimp (\ilozenge A \vee_i \ilozenge B)}
    {\infer[\ilozenge L]{x: \ilozenge(A\vee_i B) \vdash x: \ilozenge A \vee_i \ilozenge B}
    {\infer[\vee_i L]{xRy, y: A\vee_i B \vdash x: \ilozenge A \vee_i \ilozenge B}
    {\infer[\vee_i R_1]{xRy, y: A \vdash x: \ilozenge A \vee_i \ilozenge B}{\infer[\ilozenge R]{xRy, y: A \vdash x: \ilozenge A}{\infer[]{xRy, y: A \vdash y:  A}{}}}
    &\infer[\vee_i R_2]{xRy, y: B \vdash x: \ilozenge A \vee_i \ilozenge B}{\infer[\ilozenge R]{xRy, y: B \vdash x: \ilozenge B}{\infer[]{xRy, y: B \vdash y:  B}{}}}
    }}}
    \]
\item $\ka_3:(\ilozenge A \iimp \square B) \iimp \square(A \iimp B)$
    \[
    \infer[\iimp R]{\vdash x: (\ilozenge A \iimp \square B) \iimp \square(A \iimp B)}
    {\infer[\square R]{x: (\ilozenge A \iimp \square B) \vdash x: \square(A \iimp B)}
    {\infer[\iimp R]{xRy, x: (\ilozenge A \iimp \square B) \vdash y: A \iimp B}
    {\infer[\iimp L]{xRy, x: (\ilozenge A \iimp \square B), y: A \vdash y: B}
    {\infer[\ilozenge R]{xRy, x: (\ilozenge A \iimp \square B), y: A \vdash x:\ilozenge A}{\infer[]{xRy, x: (\ilozenge A \iimp \square B), y: A \vdash y: A}{}}
    &\infer[\square L]{xRy, y: A, x:\square B \vdash y: B}{\infer[]{xRy, y: A, y: B \vdash y: B}{}}}}}}
    \]
\item $\ka_4:\ilozenge\bot \iimp \bot$
    \[
    \infer[\iimp R]{\vdash x: \ilozenge\bot \iimp \bot}
    {\infer[\ilozenge L]{x: \ilozenge\bot \vdash x: \bot}
    {\infer[\bot]{xRy, y: \bot \vdash x: \bot}{}}}
    \]
\end{itemize}

\end{document}